\documentclass[aps,pra,10pt,a4paper,reprint,nofootinbib,superscriptaddress]{revtex4-2}
\usepackage[utf8]{inputenc}
\usepackage{amsfonts}
\usepackage{amssymb}
\usepackage{amsmath}
\usepackage{amsthm}
\usepackage{graphics}
\usepackage{graphicx}
\usepackage{braket}
\usepackage[colorlinks=true,linkcolor=blue,urlcolor=blue,citecolor=blue]{hyperref}
\usepackage{times,txfonts}

\newtheorem{definition}{Definition}
\newtheorem{proposition}{Proposition}
\newtheorem{remark}{Remark}

\begin{document}
\title{Quantum advantage in a unified scenario and secure detection of resources}

\author{Saronath Halder}
\affiliation{Centre for Quantum Optical Technologies, Centre of New Technologies, University of Warsaw, Banacha 2c, 02-097 Warsaw, Poland}

\author{Alexander Streltsov}
\affiliation{Institute of Fundamental Technological Research, Polish Academy of Sciences, Pawi\'{n}skiego 5B, 02-106 Warsaw, Poland}

\begin{abstract}
Quantum resources may provide advantage over their classical counterparts in communication tasks. We say this as {\it quantum advantage}. Here we consider a single communication task to study different approaches of observing quantum advantage. We say this setting as a {\it unified scenario}. Our task is described as the following. There are three parties - the Manager, Alice, and Bob. The Manager sends a value of a random variable to Alice and at the same time Bob receives some partial information regarding that value. Initially, neither Alice nor Bob knows the input of the other, which is received from the Manager. The goal of the task is achieved if and only if the value of the random variable, sent to Alice by the Manager, is identified by Bob with success probability greater than half all the time. Here non-zero error probability is allowed. However, to help Bob, Alice sends a limited amount of classical or quantum information to him (cbit or qubit). We show that the goal of the task can be achieved when Alice sends a qubit. On the other hand, a cbit communication is not sufficient for achieving the goal. Thus, it establishes quantum advantage. We further show that the optimal success probability in the overall process for a qubit communication can be higher than the optimal success probability for a cbit communication. In fact, this success probability is higher for qubit communication even if along with cbit communication Alice and Bob share randomness which allows mixing of classical strategies. Clearly, it not only establishes quantum advantage, it also demonstrates a more prominent non-classical feature. Actually, we obtain higher success probability compared to other tasks. This also suggests the experiment friendly nature of our task. For applications, we connect our task with semi-device-independence and show how our task can be useful to detect quantumness of a communication in a secure way. Moreover, within the same setting, we provide a way to detect universal coherence present in an ensemble. For a high dimensional random variable, to achieve the goal, it may require a high dimensional classical communication while it can be achieved with only a qubit communication. This establishes a large separation between the quantum and the classical communication. 
\end{abstract}
\maketitle

\section{Introduction}
The setting of communication tasks plays a crucial role to study the difference between classical and quantum communication. In a typical communication task, there are two players Alice and Bob. They receive inputs from a mediator, say, the Manager. Bob does not know about Alice's input. Likewise, Alice does not know about Bob's input. Now, Bob has to achieve a goal which is related to Alice's input. To help Bob, Alice is allowed to send some amount of information to him. What Alice does is that she tries to encode some information regarding her input into the information which she is allowed to send to Bob. Ultimately, Bob tries to achieve the goal based on the input that he receives from the Manager and also based on the communication that he receives from Alice's end. Given this setting of a communication task, there are multiple ways to study the aforesaid difference, i.e., the difference between classical and quantum communication. 

The first approach that we want to talk about is described as the following. Suppose, under the above setting, Alice is only allowed to send a `limited amount of information' which can be a quantum bit (qubit) or a classical bit (cbit). Then, it can be asked what are the tasks where the goal cannot be achieved when there is only a cbit communication but it can be achieved when there is a qubit communication. Clearly, existence of such tasks demonstrates advantage of quantum communication. In brief, we say it quantum advantage. Furthermore, it can be analyzed if the optimal success probability to achieve the goal in the quantum case is higher than that in the classical case. This is particularly important as there is a fundamental bound which tells us that from a qubit only a cbit of information is extractable with certainty \cite{Holevo73} (see also \cite{Frenkel15} in this regard). This is due to the fact that a qubit state can be prepared in infinitely many ways but only two states can be distinguished perfectly. In spite of the existence of such a limitation, there are tasks where it is possible to achieve better success probability in acquiring the goal when there is a qubit communication. For example, one can consider the task of a random access code (RAC) and its quantum version \cite{Wiesner83, Ambainis99, Ambainis02, Hayashi06-1, Tavakoli15, Vaisakh21}. For other relevant works see Refs.~\cite{Mukherjee21, Das21, Kerppo22, Kanjilal23, Chakraborty23, Saha23, Abhyoudai23, Heinosaari24-1, Patra24, Scala24, Ding24, Heinosaari24} and the references therein. We now describe one of the main motivations to consider this approach. Suppose, it is not possible to achieve the goal of a given task with certainty but it is possible to get quantum advantage in the task. Then, the question of interest is that up to what extent one can push the success probability to achieve the goal of the task. Can it be very close to one? Obviously, to answer this question, one has to consider the communication tasks where it is not only possible to obtain quantum advantage but it is also possible to achieve the goals with better success probabilities than that of the tasks associated with quantum random access codes.  

From the above discussion, it is quite understandable why the advantage of a qubit communication is surprising. Raising this surprise, we now want to talk about the tasks where there is a significant advantage of quantum communication over its classical counterpart. Here comes the second approach which deals with a question that belongs to the theory of communication complexity \cite{Yao79}. This question can be described as the following. What is the minimum amount of classical or quantum information, necessary for Alice to send it to Bob such that he can achieve the goal of a given task? So, here we can consider those tasks where the goal can be achieved if Alice is allowed to send a qubit to Bob. On the other hand, in such a task, to achieve the goal, it might be necessary that Alice has to send a large number of cbits. In this context, one can have a look into Refs.~\cite{Massar01, Buhrman01, Wolf03, Galvao03, Perry15, Halder24, Rout23, Heinosaari24-1}. Here our interest lies in unifying this approach with the previous one through the construction of a single communication task. The main motivations behind this unification are: (a) if we consider a single communication task where several approaches can be studied then, it might be easier to understand different features of `quantum advantage', (b) using limited resources, one can experimentally demonstrate different directions of quantum advantage. The second point, i.e., (b) is particularly important because of the hardness to build quantum technologies \cite{Bruzewicz19, Pelucchi22, Cerezo22}. Furthermore, such a single communication task might be useful from several other aspects. In the following, we describe our communication task and corresponding aspects gradually.

As mentioned, in this work we consider aforesaid approaches in a single task. We describe this as a {\it unified scenario}. This can be better understood through a game (see Fig.~\ref{fig1}) where there are three parties: the Manager, Alice, and Bob. The Manager sends a value of a random variable to Alice and at the same time Bob receives partial information regarding that value. Bob does not know which value is sent to Alice, likewise, Alice also does not know what information is sent to Bob. The goal of the task can be achieved if and only if the value of the random variable can be identified by Bob with success probability greater than $1/2$ no matter what information he receives from the Manager or which value of the random variable is sent to Alice by the Manager. We mention that when Bob tries to identify the value of the variable, non-zero probability of error is allowed. Clearly, this setting is different from a conclusive setting \cite{Halder24} where there is no scope to commit error. However, to help Bob, Alice sends a limited amount of classical or quantum information to Bob. We show that when Alice sends a qubit the task can be accomplished. On the other hand, a cbit communication is not sufficient for accomplishing the task. We further show that the optimal success probability in case of a qubit communication can be higher than that in case of a cbit communication. In fact, this success probability is higher for qubit communication even if along with cbit communication Alice and Bob share randomness which allows mixing of classical strategies. Thus, we argue that the non-classical nature is more prominent in the present task. Actually, we obtain higher success probability in the quantum case compared to quantum random access code (QRAC). Therefore, our task is more experiment friendly. Interestingly, in our task, for a high dimensional random variable, it may require to communicate a high dimensional classical information to achieve the goal of the task while it can be achieved via only a qubit communication. This clearly establishes a large separation between the quantum and the classical communication. However, we mention that there are papers where arbitrary quantum advantage is reported \cite{Galvao03, Perry15, Rout23}. But in our case we consider a single setting to observe several incidents. This is certainly one of the key points of our work. 

\begin{figure}[t!]
\centering
\includegraphics[scale=0.28]{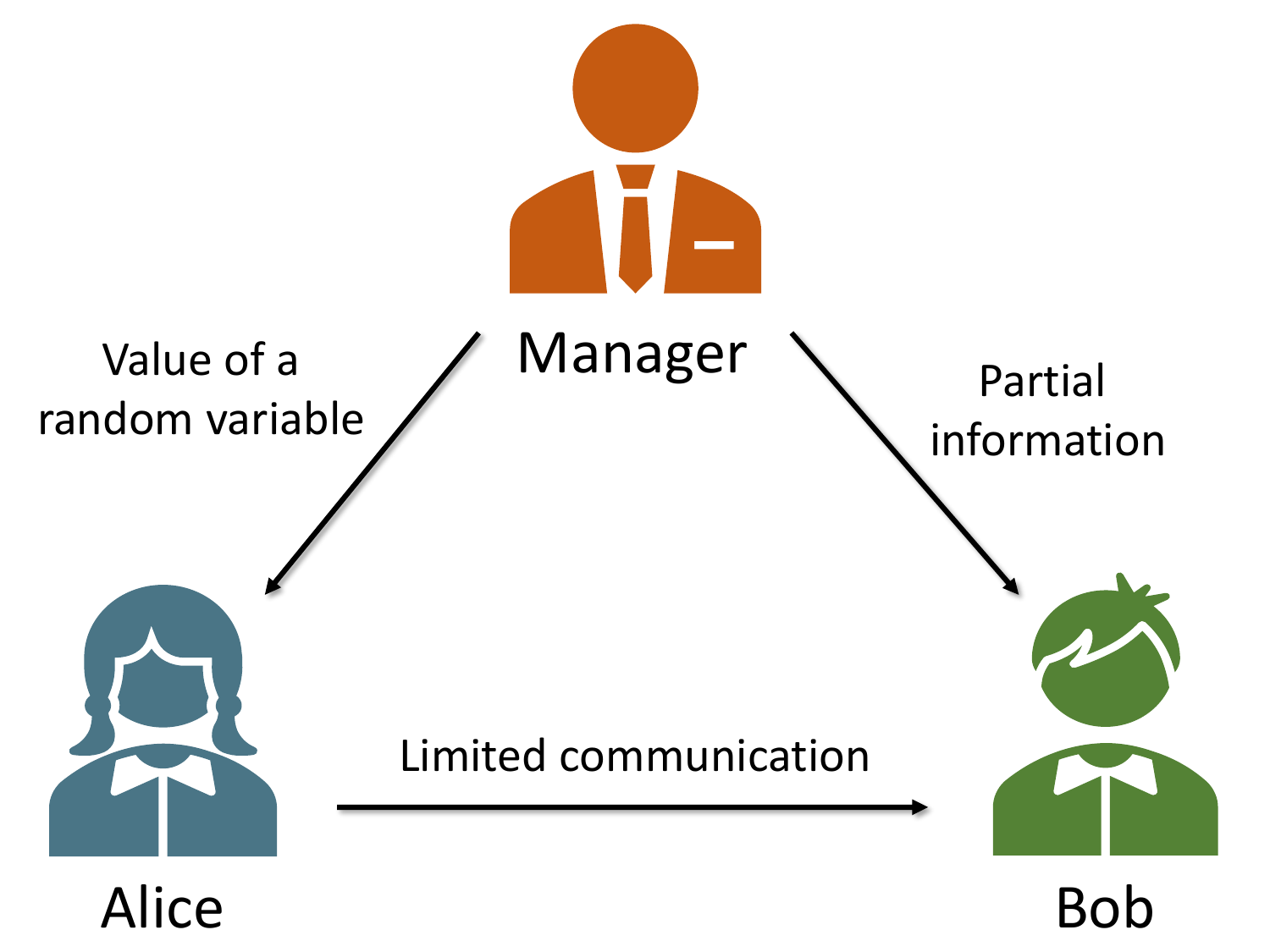}
\caption{Schematic diagram for our communication task.}\label{fig1}
\end{figure}

Now, to build quantum technologies, it is very important to understand the difference between classical and quantum information. In particular, for the advancement of such technologies, it is necessary to understand the advantage that one can obtain using quantum resources over their classical counterparts. This advantage, in communication tasks, mainly comes from the quantumness of communication. It may have connection with another resource like quantum coherence \cite{Streltsov17}. Certification of security \cite{Scarani09, Brunner14, Herrero17, Xu20, Supic20, Uola20, Portmann22} is also an important aspect of quantum technologies. In many cases, researchers talk about device-independent (DI) strategies as they deal with the secure ways of accomplishing tasks. In our case, to exhibit quantum advantage, quantum coherence plays a very important role. Apart from examining this role, we ask if our communication task can lead to any application which will improve our understanding regarding quantum resources. This can be regarded as another motivation of our work. To execute this purpose, we first connect our communication task with semi-device-independence. Then, we explore how our task can be employed to detect quantumness of communication. We further provide a sufficient condition to detect universal coherence of an ensemble using our communication task and describe corresponding methodology.

In the following we first present our main findings step by step and then we draw the conclusion.

\section{Our task}
The Manager sends a value of a random variable $X$ to Alice. This value is not known to Bob. At the same time, some partial information regarding this value is sent to Bob. Again, this partial information is not known to Alice. The goal of the task is: Bob has to identify the value of $X$, which is sent to Alice. In fact, this should be done with a probability $p>1/2$ all the time, i.e., does not matter which value is sent to Alice or what information is received by Bob. However, to help Bob, a `limited amount of information' is sent to him by Alice. In the following we first consider that $\dim X$ = $d$ (say) = 3 and describe the task more precisely. We note here that to achieve a probability $p = 1/2$, it is not required to send any information to Bob as this success probability can be achieved through random guess. Therefore, it is important to fix a probability $p>1/2$. We mention that when Bob tries to identify the value of $X$, which is sent to Alice by the Manager, non-zero error probability is allowed. When $d=3$, apparently, the success probability of identifying the value of $X$ via random guess should be 1/3. However, here Bob receives additional information from the Manager and this is why the sucess probability through random guess can be 1/2 here.

\subsection{Three dimensional random variable}
The Manager sends a value of $X$ to Alice randomly. It is denoted by $x_i\in\{x_1, x_2, x_3\}$. This set is known to both Alice and Bob but the particular $x_i$ which is sent to Alice by the Manager, is not known to Bob. At the same time, the Manager sends the information of a random set to Bob. The cardinality of the set is two here and this set must contain the value of $X$, which is sent to Alice, along with some other value. These information are known to both Alice and Bob but the particular information which is sent to Bob by the Manager, is not known to Alice. So, we say that the Manager sends Bob `$j$' where this `$j$' is associated with $S_j$. Based on the dimension of $X$ and the cardinality of $S_j$, it is possible to define three sets $S_1 = \{x_1, x_2\}$, $S_2 = \{x_2, x_3\}$, and $S_3 = \{x_3, x_1\}$. These definitions are also known to both Alice and Bob. Since, the value of $X$, which is sent to Alice by the Manager, is randomly chosen, the information of which set is sent to Bob is also random. Now, to help Bob in identifying $x_i$ with a success probability $p>1/2$, Alice is allowed to send either a cbit or a qubit. The goal of the task can be achieved, i.e., the game can be won if and only if Bob is able to identify $x_i$ with $p>1/2$, does not matter what is the value of `$i$' or which `$j$' value Bob receives. Let us take one example for more clarity. Suppose, the Manager sends $x_1$ to Alice. Then, at the same time Bob receives either `1' or `3' randomly. We now present the first proposition of the paper.

\begin{proposition}\label{prop1}
The communication task which is described in the above does not have any classical counterpart, i.e., the game cannot be won if Alice is allowed to send only a cbit.
\end{proposition}

\begin{proof}
Before we provide the main argument of the proof it is important to note the following. Ultimately, Bob distinguishes between two values of the random variable. Apparently, information of two values of $X$ can be sent through a cbit. But Alice does not know which two values of $X$, Bob distinguishes between. In fact, a value of $X$ is sent to Alice randomly and thus, the information regarding a set is also randomly sent to Bob. Clearly, all values of $X$ are probable when Alice thinks of Bob distinguishing the values of the random variable $X$. So, the partial information which is sent to Bob is useful to him for decoding the information but it is not useful for Alice to encode the information.

The rest of the proof is based on the following. Alice must encode some property associated with the values of $X$ into a function. This function must output either `0' or `1' because Alice is only allowed to send a cbit to Bob. Since, $\dim X$ = 3, if the function outputs 0 or 1, then there are at least two values for which the function outputs the same. This implies that there exists at least one `$j$' and thereby an $S_j$, for which the function outputs the same value corresponding to the values of $X$ contained in $S_j$. Thus, these values cannot be distinguished with $p>1/2$. These complete the proof.
\end{proof}

{\it Optimal success probability in the classical case}: From the above, it is clear that there exists one `$j$' value, for which it is not possible to distinguish the values with $p>1/2$. So, in this case maximal success probability is $1/2$. For other two `$j$' values the success probabilities can be 1. Such a situation can be achieved through the following encoding: Alice sends `0' to Bob if she receives either $x_1$ or $x_2$, otherwise, she sends `1' to Bob. These imply that in the classical case maximum success probability is the average of three success probabilities corresponding to three `$j$' values and it is given by- $(1/3)\times 1 + (1/3)\times 1 + (1/3)\times(1/2)$ = 5/6 = 0.83 approx.

Next, we are interested to show that if Alice is allowed to send a qubit to Bob, then, not only the game can be won but they can also achieve a better success probability compared to the classical case. 

\begin{proposition}\label{prop2}
If Alice is allowed to send a qubit to Bob, then the goal of the task can be achieved, i.e., the game can be won.
\end{proposition}

\begin{proof}
To prove the proposition, it is sufficient to provide an encoding and a decoding strategy via which the game can be won. Alice's encoding strategy is given as the following:

\begin{equation}\label{eq4}
x_1 \rightarrow \ket{0},~ x_2 \rightarrow \frac{1}{2}(\ket{0}-\sqrt{3}\ket{1}),~ x_3 \rightarrow \frac{1}{2}(\ket{0}+\sqrt{3}\ket{1}).
\end{equation}
These states together form the so-called {\it trine ensemble} (for details see Ref.~\cite{Barnett09} and references therein). For this encoding, Bob has to distinguish between two nonorthogonal states for all values of `$j$'. Corresponding optimal probabilities can be found through the Helstrom bound \cite{Helstrom76}. It is easy to see that for every `$j$' value the probability of success is greater than 1/2. Thus, the game can be won.
\end{proof}

We now want to talk about the maximum success probability which can be achieved through the encoding given in (\ref{eq4}). Here, for all `$j$' values, Bob distinguishes between two nonorthogonal states. Corresponding success probabilities can be calculated form the well known Helstrom bound \cite{Helstrom76}. This is given by-
\begin{equation}
p = \frac{1}{2}\left(1+\sqrt{1-|\langle\phi|\phi^\prime\rangle|^2}\right).
\end{equation}
where `$p$' is the probability of success and $\{\ket{\phi}, \ket{\phi^\prime}\}$ are the states (equally probable) to be distinguished. For {\it trine ensemble}, it is known that $|\langle\phi|\phi^\prime\rangle|$ = $1/2$ for any pair of states from the ensemble. So, for the encoding given in (\ref{eq4}), the average success probability becomes $(1/3)(p+p+p)$ = $p$ = 0.93 approx. Clearly, this is better than the optimal probability of success in the classical case. We now have a remark regarding our communication task.

\begin{remark}
We establish advantage of quantum communication in the context of the present task. In fact, we exhibit a more prominent non-classical nature of this task. 
\end{remark}

From Proposition \ref{prop1}, we find that a `limited amount of classical information' is not sufficient for achieving the goal of the task. On the other hand, the quantum version of this `limited amount of information' is sufficient for achieving the goal of the present task. This is given in Proposition \ref{prop2}. Thus, we establish advantage of quantum communication. We mention whether it is possible to extract any advantage of quantum communication, depends on the goal we set. From Proposition \ref{prop2}, it is only evident that the winning condition can be achieved through qubit communication. This cannot guarantee that the average success probability in the overall process is larger in the quantum case compared to that of the classical case. In this regard, one should have a look into Ref.~\cite{Halder24}. Nevertheless, here we see that the average success probability in the overall process in the quantum case can be higher than the optimal success probability in the classical case. 

Note that for both propositions (\ref{prop1} and \ref{prop2}), Alice and Bob consider only pure strategies for encoding-decoding. We then assume that along with classical communication there is shared randomness between Alice and Bob. We also assume that this randomness allows the mixing of pure classical strategies. But via this mixing of classical strategies it is not possible to have a better success probability than 5/6. The reason can be described as the following. The success probability 5/6 is basically the upper bound of success probability in the classical case. So, corresponding strategy is the best possible pure strategy. Clearly, if one takes convex combination of such strategies then the success probability cannot be higher than 5/6. In this sense, here quantum communication is more advantageous than classical communication along with shared randomness. This is obviously a stronger demonstration of quantum advantage. 

Therefore, we argue that in the present communication game, we observe a more prominent non-classical nature of the task. However, it is possible to prove that the optimality of the success probability in the quantum case can be achieved through the encoding of (\ref{eq4}). 

{\it Optimality in the quantum case}: In general, we consider an encoding of the following form: $x_i\rightarrow\ket{\phi_i}$, $\forall i =1,2,3$. From the set $\{\ket{\phi_i}\}_i$, if we consider any two states then they must be linearly independent, otherwise, corresponding to a `$j$' value, the success probability becomes $1/2$. This is not desired. So, if we take $\ket{\phi_1} = \ket{\phi}$ then $\ket{\phi_2} = a\ket{\phi}+b\ket{\phi^\perp}$ and $\ket{\phi_3} = c\ket{\phi}+d\ket{\phi^\perp}$, where $\langle\phi|\phi^\perp\rangle = 0$, $|a|^2+|b^2|$ = $|c|^2+|d|^2$ = 1, $a/b\neq c/d$, $|b|, |d|>0$. $\ket{\phi_3}$ can also be written as $(ca^\ast+db^\ast)\ket{\phi_2}+(cb-da)\ket{\phi_2^\perp}$, where $\ket{\phi_2^\perp} = b^\ast\ket{\phi}-a^\ast\ket{\phi^\perp}$, $|ca^\ast+db^\ast|^2+|cb-da|^2=1$. We should remember that $a,b,c,d$ are complex numbers and the symbol `$\ast$' suggests complex conjugation of some complex number. We now define $p_e^j$ $\forall j=1,2,3$. These are the probabilities of error corresponding to each `$j$' value. They are given by-
\begin{equation}
\begin{array}{l}
p_e^1=\frac{1}{2}\left(1-\sqrt{1-|\langle\phi_1|\phi_2\rangle|^2}\right)=\frac{1}{2}\left(1-|b|\right),\\[1 ex]

p_e^2=\frac{1}{2}\left(1-\sqrt{1-|\langle\phi_2|\phi_3\rangle|^2}\right)=\frac{1}{2}\left(1-|cb-da|\right),\\[1 ex]

p_e^3=\frac{1}{2}\left(1-\sqrt{1-|\langle\phi_3|\phi_1\rangle|^2}\right)=\frac{1}{2}\left(1-|d|\right).
\end{array}
\end{equation}
So, the average probability of error $p_e$ = $(1/3)(p_e^1+p_e^2+p_e^3)$. This is given by-
\begin{equation}
p_e = \frac{1}{2}-\frac{1}{6}\left(|b|+|cb-da|+|d|\right)
\end{equation}
We next aim at the minimization of $p_e$ which can be done by maximizing the quantity $(|b|+|cb-da|+|d|)$. The maximum value of this quantity can be found by applying {\it triangle inequality} in the following way:
\begin{equation}
\begin{array}{l}
|b|+|cb-da|+|d| = |b|+|cb+(-d)a|+|d|\\[1 ex]

\leq |b|+|c||b|+|-d||a|+|d|\\[1 ex]

= (1+|c|)|b|+(1+|a|)|d|.
\end{array}
\end{equation}
In this way, we arrive to a maximization problem. However, we can consider the maximization problem for general states $\ket{\phi_2}$ and $\ket{\phi_3}$, i.e., $0\leq |a|,|b|,|c|,|d|\leq 1$ and let us see what is the solution.
\begin{equation*}
\begin{array}{l}
\mbox{maximize}~~(1+|c|)|b|+(1+|a|)|d|\\[1 ex]

\mbox{subject to}~~0\leq |a|,|b|,|c|,|d|\leq 1,~ |a|^2+|b|^2 = |c|^2+|d|^2 = 1.
\end{array}
\end{equation*}
This problem can be simplified in the following way:
\begin{equation*}
\begin{array}{l}
\mbox{maximize}~~(1+|a|)\sqrt{1-|c|^2}+(1+|c|)\sqrt{1-|a|^2}\\[1 ex]

\mbox{subject to}~~0\leq |a|,|c|\leq 1.
\end{array}
\end{equation*}
So, basically we have to look for two positive numbers $|a|$ and $|c|$ between the interval $[0,1]$ such that the above mentioned quantity is maximized. But even if we get a mathematical solution, the question is its achievability in a practical situation. It is easy to check that when $|a|=|c|=1/2$ the quantity $(1+|a|)\sqrt{1-|c|^2}+(1+|c|)\sqrt{1-|a|^2}$ is maximized. Notice that this solution is achievable in a practical situation via the encoding of (\ref{eq4}). So, for this encoding $p_e$ becomes minimum and thereby, maximization of average success probability in the overall process occurs. This maximal value is given by- (1/2+$\sqrt{3}$/4) = 0.93 approx.

\subsection{Application of our communication task}
Here we discuss about a possible application of our communication task. In particular, we want to discuss how the present task can be employed to detect {\it quantumness} in a semi-DI way (for semi-device-independence, one can have a look into Ref.~\cite{Pawlowski11}). 

In the usual device-independent (DI) setup, one (say, the Manager) who wish to detect some property of the physical system, is not required to understand the underlying theory of the device which is in use. For a given system, the Manager is just required to get rid of the inputs and the outputs. From these data, if the Manager can conclude something, which is required to achieve the goal, without committing any error, then we say that the task is completed in a DI way. For example, for a given quantum state if a person collect information of inputs and outputs of measurements performed in the so-called Bell scenario \cite{Brunner14}, then that person can construct the conditional probabilities. Employing these probabilities if it is possible to show violation of any Bell's inequality, then we can say that detection of entanglement within the given state is done in a DI manner. In this case, the person knows nothing about the mechanisms of the devices which are in use (see \cite{Brunner14, Supic20} and the references therein). That person is only given a state which is either separable or entangled. We mention that DI methods are based on nonlocality which strongly suggests its connection with entanglement based protocols. Nevertheless, our protocol does not include entanglement and thus, it is natural to consider its connection with semi-DI method.

In our case, we consider black boxes for Alice and Bob. Alice's black box is a preparation black box and Bob's one is a measurement black box. The classical or quantum information which is sent by Alice to Bob is a two-level information. So, the dimension of the communicating information is given. Suppose, the Manager wish to learn if Alice sends a classical bit or quantum bit to Bob. Here the Manager does not have to understand fully the underlying theory of the communication task. A semi black-box scenario can be considered where the inputs and the outputs are only considered. For example, Alice's input is `$x_i$' and Bob's input is `$j$'. Here we assume that the `$j$' values are determined by a classical variable. This is because for a given $x_i$, all `$j$' values are not allowed. Alice feeds a cbit or a qubit to Bob. The Manager wishes to detect the classical or non-classical property of this communication. For this purpose, all that the Manager knows is - Alice either sends a qubit or a cbit. It is also assumed how the `$j$' values are defined. Preparations of Alice and measurements of Bob are not characterized. See Fig.~\ref{fig2} for this application.

\begin{figure}
\centering
\includegraphics[scale=0.25]{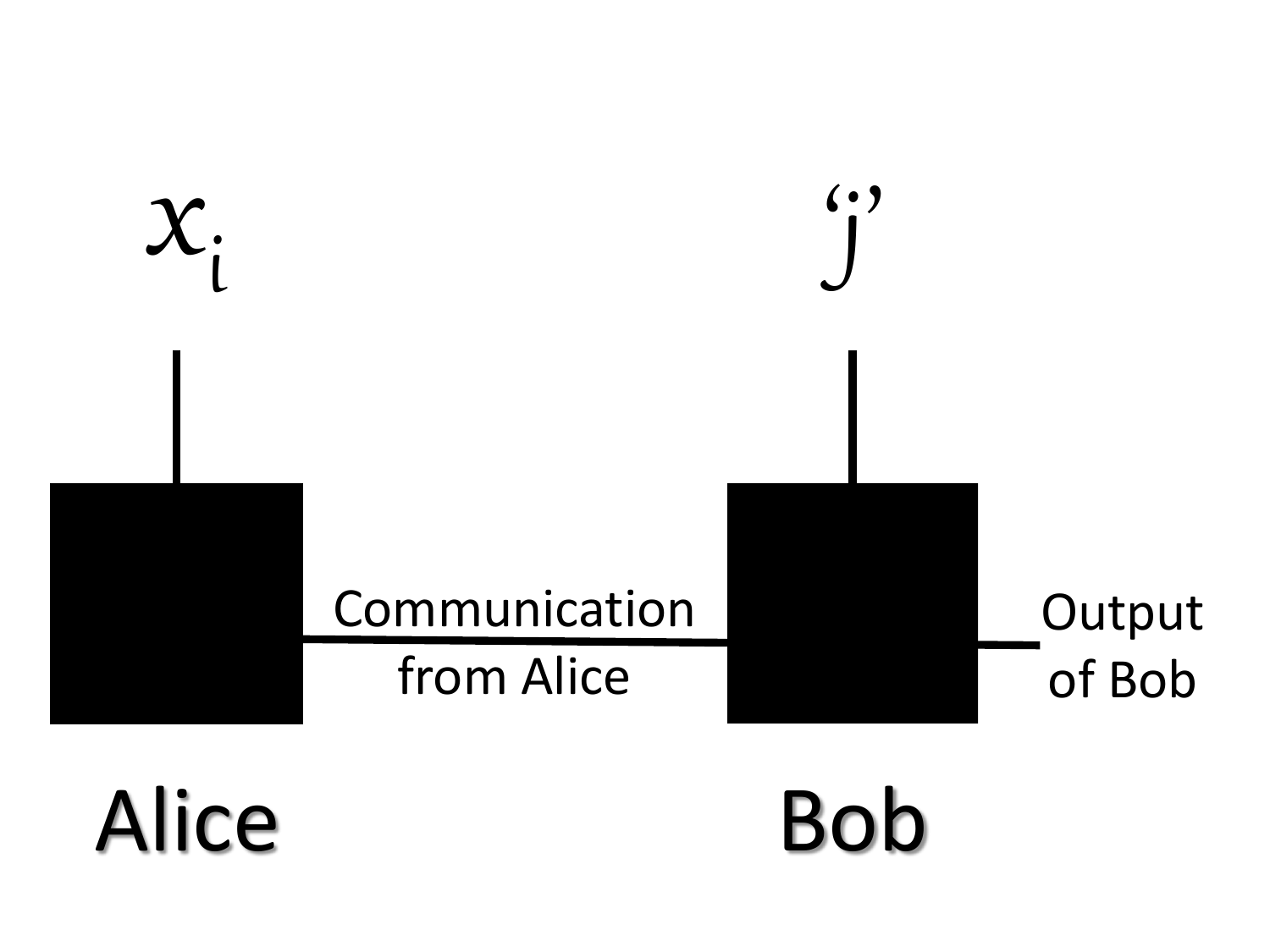}
\caption{Schematic diagram to detect non-classicality of Alice's communication to Bob in a semi-device-independent way.}\label{fig2}
\end{figure}

After completion of each round, Bob announces his conclusion regarding $x_i$. So, based on the inputs of Alice and Bob and the output of Bob, the Manager constructs the conditional probabilities $p^s(i|x_i,j)$ after repeating the procedure for many rounds. These are the probabilities of successfully identifying `$i$' when a random `$x_i$' and `$j$' is given. If Alice communicates a classical bit then it is known that 
\begin{equation}
\sum_{i,j} p^s(i|x_i,j) \leq \frac{5}{6}.
\end{equation}
The above inequality is due to the optimal success probability achievable in the classical case. If the above inequality is violated then Alice sends a qubit to Bob. We further want to prove that if individual values of $p^s(i|x_i,j)$ is greater than $1/2$ $\forall i,j$, then the ensemble which is used by Alice for encoding in the quantum case must contain quantum coherence \cite{Streltsov17}. We define this coherence as `universal'. However, we start with the following definition:

\begin{definition}{[Coherence of an ensemble]} We consider an ensemble of pure states $\{\ket{\phi_i}\}_i$. If it is not possible to find any basis under which all states $\ket{\phi_i}$ are diagonal then we say that the ensemble contains coherence.
\end{definition}

\noindent
Notice that in the above definition we have removed the basis dependence of coherence. This is because if an ensemble contains coherence according to the above definition then it does so with respect to all bases. Therefore, we say that the coherence of the present kind is a {\it universal} one. We now proceed to prove the following:

\begin{proposition}
We assume that Alice uses the encoding $x_i\rightarrow\ket{\phi_i}$. Now, the ensemble $\{\ket{\phi_1}, \ket{\phi_2}, \ket{\phi_3}\}$ contains universal coherence if  $p^s(i|x_i,j)>1/2$ $\forall i,j = 1,2,3$.
\end{proposition}

\begin{proof}
All three states \{$\ket{\phi_1}$, $\ket{\phi_2}$, $\ket{\phi_3}$\} must be different, i.e., pairwise linearly independent, otherwise, corresponding to at least one `$j$' value $p^s(i|x_i,j)$ becomes $1/2$. This is not desired. We first consider two states $\ket{\phi_1}$ and $\ket{\phi_2}$. They are linearly independent and thus, they form a basis for a two-dimensional Hilbert space. So, $\ket{\phi_3}$ can be written as $a\ket{\phi_1}+b\ket{\phi_2}$, where $a,b$ can be complex numbers such that $|a|,|b|>0$. These complex numbers are chosen in such a way that $\ket{\phi_3}$ is a valid quantum state. Clearly, $\ket{\phi_1}$ and $\ket{\phi_2}$ might be diagonal in a basis but with respect to the same basis $\ket{\phi_3}$ must not be diagonal. Since, $\ket{\phi_1}$ and $\ket{\phi_2}$ can be any state, it is exhibiting that the ensemble $\{\ket{\phi_1}, \ket{\phi_2}, \ket{\phi_3}\}$ always contains coherence if we want to satisfy the inequality $p^s(i|x_i,j)>1/2$ $\forall i,j = 1,2,3$.
\end{proof}

\noindent
Using the above proposition, it is clear why it is possible to detect coherence of an ensemble via our communication task. Furthermore, considering the semi-DI method, we can detect universal coherence of any ensemble, which is used by Alice for encoding in the quantum case. 

\subsection{Four dimensional random variable}
We first define the task for $d=4$ briefly following the $d=3$ case. A value $x_i\in$\{$x_1$, $x_2$, $x_3$, $x_4$\} of the random variable $X$ is sent to Alice. At the same time, Bob receives a `$j$' value, where `$j$' is associated with $S_j$. In this case also, $S_j$ contains two values of $X$, the original value which is sent to Alice along with some other value. To help Bob in identifying the value of $X$, Alice is allowed to send either a cbit or a qubit. Possible `$S_j$'s are: $S_1 = \{x_1, x_2\}$, $S_2 = \{x_2, x_3\}$, $S_3 = \{x_3, x_4\}$, $S_4 = \{x_4, x_1\}$, $S_5 = \{x_1, x_3\}$, $S_6 = \{x_2, x_4\}$. So, the game goes here in the following way. The Manager sends $x_i\in$ \{$x_1$, $x_2$, $x_3$, $x_4$\} to Alice. At the same time, a `$j$' value is sent to Bob by the Manager, where $j=1,2,\dots,6$. This implies that Bob is instructed that the value which is sent to Alice, belongs to $S_j$. However, it is important to mention that for a given $x_i$, not all values of `$j$' are allowed. For example, if $x_1$ is sent to Alice, the Manager sends any one `$j$' value, where $j\in\{1,4,5\}$. We also mention that Bob does not know which value of $X$ is sent to Alice, likewise, which `$j$' value is sent to Bob is not known to Alice. Finally, with the help of Alice, Bob tries to identify the value of $X$ which is sent to Alice and this should be done by reducing the error probability. The task can be accomplished if and only if Bob is able to identify the value of $X$, which is sent to Alice by the Manager, with a success probability greater than half, no matter what information he receives or which value is sent to Alice.

{\it Classical case}: We first ask if the task can be accomplished by sending a cbit. Here the dimension of $X$ is increased. The rest goes in the same way as in the previous case, i.e., Alice sends a cbit or a qubit, Bob receives a `$j$' value, and $S_j$ contains only two values. Clearly, with increasing dimension, it cannot be the case that the task is accomplished by sending a cbit. A formal proof can be obtained via the same line of arguments as given in the proof of the Proposition \ref{prop1}.

The second thing we ask is about the maximum success probability achievable through one cbit of communication from Alice's side to Bob. Following the above argument, we can infer that with increasing dimension of $X$, one cannot achieve a better success probability compared to $d=3$ case. Interestingly, here the success probability $p=0.83$ can be achieved through the following encoding $x_1, x_2 \rightarrow 0$ and $x_3,x_4 \rightarrow 1$. Thus, for $d=4$ also, the maximum success probability is $p=0.83$ approx.

{\it Quantum case}: To prove that the task can be accomplished when Alice is allowed to send a qubit to Bob, they can plan for the following encoding strategy.
\begin{equation}\label{eq7}
\begin{array}{c}
x_1 \rightarrow \ket{0},~ x_2 \rightarrow \frac{1}{\sqrt{3}}(\ket{0}-\sqrt{2}\ket{1}),\\[1 ex] x_3 \rightarrow \frac{1}{\sqrt{3}}(\ket{0}-\sqrt{2}e^{2\pi i/3}\ket{1}),\\[1 ex] x_4 \rightarrow \frac{1}{\sqrt{3}}(\ket{0}-\sqrt{2}e^{-2\pi i/3}\ket{1}).
\end{array}
\end{equation}
It is easy to check that using the above encoding the parties can achieve $p>1/2$ $\forall i,j$ and an average success probability ($1/2 + 1/2\times\sqrt{2/3}$) = 1/2 + 1/$\sqrt{6}$ = 0.908 approx. Clearly, they can accomplish the task as well as they can also achieve an average success probability which is better compared to the classical case. We mention that the states in the above encoding strategy form the so-called {\it tetrad} ensemble (see \cite{Barnett02} and references therein).

We now provide a detailed derivation regarding the general encoding strategy for $d=4$. If Alice considers the encoding $x_i\rightarrow\ket{\phi_i}$ $\forall i=1,2,3,4$, then, the states $\{\ket{\phi_i}\}$ must be pairwise linearly independent. For this purpose, we choose the following: $\ket{\phi_1} =\ket{\phi}$, $\ket{\phi_2}=a\ket{\phi}+b\ket{\phi^\perp}$, $\ket{\phi_3}=c\ket{\phi}+d\ket{\phi^\perp}$, $\ket{\phi_4}=e\ket{\phi}+f\ket{\phi^\perp}$, $a,b,c,d,e,f$ are complex numbers, $|a|^2+|b|^2$ = $|c|^2+|d|^2$ = $|e|^2+|f|^2$ = 1. The numbers $a,b,c,d,e,f$ should be chosen in such a way that $\ket{\phi_2}$, $\ket{\phi_3}$, $\ket{\phi_4}$ become different states, i.e., they must be pairwise linearly independent. We now consider the error probabilities $p_e^j$ for each `$j$' values. The probabilities $p_e^1$, $p_e^4$, $p_e^5$ are given by-
\begin{equation}
\begin{array}{l}
p_e^1=\frac{1}{2}\left(1-\sqrt{1-|\langle\phi_1|\phi_2\rangle|^2}\right)=\frac{1}{2}(1-|b|),\\[1 ex]

p_e^4=\frac{1}{2}\left(1-\sqrt{1-|\langle\phi_1|\phi_4\rangle|^2}\right)=\frac{1}{2}(1-|f|),\\[1 ex]

p_e^5=\frac{1}{2}\left(1-\sqrt{1-|\langle\phi_1|\phi_3\rangle|^2}\right)=\frac{1}{2}(1-|d|).
\end{array}
\end{equation}
Clearly, $|b|, |d|, |f|$ must be nonzero. Next, we think about $p_e^2$ and $p_e^6$. But before we calculate these probabilities, we express $\ket{\phi_3}$ and $\ket{\phi_4}$ in terms of $\ket{\phi_2}$ and $\ket{\phi_2^\perp}$, where $\ket{\phi_2^\perp}$ = $b^\ast\ket{\phi}-a^\ast\ket{\phi^\perp}$, $a^\ast, b^\ast$ are complex conjugates of $a,b$ respectively. The expressions are given by- $\ket{\phi_3}$ = $(ca^\ast+db^\ast)\ket{\phi_2}+(cb-da)\ket{\phi_2^\perp}$ and $\ket{\phi_4}$ = $(ea^\ast+fb^\ast)\ket{\phi_2}+(eb-fa)\ket{\phi_2^\perp}$, where $|ca^\ast+db^\ast|^2+|cb-da|^2$ = 1 = $|ea^\ast+fb^\ast|^2+|eb-fa|^2$. Now, the expressions of $p_e^2$ and $p_e^6$ are given by-
\begin{equation}
\begin{array}{l}
p_e^2=\frac{1}{2}\left(1-\sqrt{1-|\langle\phi_2|\phi_3\rangle|^2}\right)=\frac{1}{2}(1-|cb-da|),\\[1 ex]

p_e^6=\frac{1}{2}\left(1-\sqrt{1-|\langle\phi_2|\phi_4\rangle|^2}\right)=\frac{1}{2}(1-|eb-fa|).
\end{array}
\end{equation}
Here also, $|cb-da|, |eb-fa|$ must be nonzero. Finally, we want to calculate $p_e^3$. Before that we express $\ket{\phi_4}$ in terms of $\ket{\phi_3}$ and $\ket{\phi_3^\perp}$, where $\ket{\phi_3^\perp}$ = $d^\ast\ket{\phi}-c^\ast\ket{\phi^\perp}$, $c^\ast,d^\ast$ are obtained by taking complex conjugation of $c,d$ respectively. So, $\ket{\phi_4}$ can be rewritten as $(ec^\ast+fd^\ast)\ket{\phi_3}+(ed-fc)\ket{\phi_3^\perp}$, $|ec^\ast+fd^\ast|^2+|ed-fc|^2$ = 1. Therefore, $p_e^3$ is given by-
\begin{equation}
p_e^3 = \frac{1}{2}\left(1-\sqrt{1-|\langle\phi_3|\phi_4\rangle|^2}\right) = \frac{1}{2}(1-|ed-fc|),
\end{equation}
where $|ed-fc|$ is nonzero. So, average probability of error is $p_e$ = $(1/6)\sum_{j=1}^6 p_e^j$ and it is given by-
\begin{equation}\label{eq11}
\small
\begin{array}{l}
p_e = \frac{1}{2}-\frac{1}{12}(|b|+|d|+|f|+|cb-da|+|eb-fa|+|ed-fc|)\\[1 ex]
= \frac{1}{2}+\frac{1}{12}(|b|+|d|+|f|)-\frac{1}{12}(\Delta_1+\Delta_2+\Delta_3),
\end{array}
\end{equation}
\normalsize
where $\Delta_1$ = $|b|+|d|+|cb-da|$, $\Delta_2$ = $|b|+|f|+|eb-fa|$, and $\Delta_3$ = $|d|+|f|+|ed-fc|$. Clearly, to reduce $p_e$, one has to increase the quantities $\Delta_i$, $\forall i=1,2,3$. From the $d=3$ case, it is already known that $\Delta_i$ can be maximized when corresponding states in encoding form a trine ensemble. So, the whole problem boils down to the following: in $d=4$ case there are four states from which if we pick any three states then they must form the trine ensemble. However, it is the tetrad ensemble from which if we pick any three states then they are equidistant. Thus, the best encoding strategy that we find is due to the states of the tetrad ensemble. It is also important to add the following here. If we consider an arbitrary dimension (finite) for the random variable, then also we can express the average probability of error like we did in (\ref{eq11}). So, using the same logic as explained in the above, the best encoding strategy that we find is due to a set of states where the states are equidistant.

Nevertheless, an important feature of our task is that one can achieve better success probability in the quantum case compared to existing task(s). For example, we compare with the task of Quantum random access code (QRAC). One can see that when Alice is given a two-bit string and she is allowed to send a qubit to Bob, the optimal success probability is 0.85. This is clearly less than the present case of $d=4$. Now, given a practical situation, it is always important to have a success probability which is greater. This is because bigger success probability suggests relatively easier experimental verification. Therefore, we argue that the present task is relatively experiment friendly. When the dimension of $X$ is five or six, in these cases also, it is possible to show that the success probability is not only higher in the quantum case compared to the classical case, it is even higher than that of QRAC.

We mention that in our case the inputs of Alice and Bob are correlated. For example, when $d=4$ if $x_1$ is sent to Alice then $j\neq2,3,6$. Now, this correlation is the key factor why it is possible to achieve better success probability here compared to QRAC. In this context, we refer to Refs.~\cite{Hall10, Barrett11} where it is discussed that correlated inputs can provide better results compared to uncorrelated inputs in the nonlocality scenario.

\subsection{A large quantum-classical separation}
We now consider that the random variable is $d$-dimensional and following the $d=3$ case, we briefly provide the description of the game. So, the value of $X$ is given by- $x_i\in\{x_1, x_2, \dots, x_d\}$. A value of $X$ is sent to Alice randomly. At the same time, Bob receives a `$j$' value, where $j$ is associated with $S_j$. $S_j$ contains only two values of $X$, the original value which is sent to Alice and some other value. By sending a `$j$' value to Bob, he is basically instructed that the value, which is sent to Alice, belongs to $S_j$. There  are total $^dC_2$ = $d(d-1)/2$ ways to define various $S_j$s. However, for a given $x_i$, there are $d-1$ allowed values of `$j$'. We mention that Bob does not know which value of $X$ is sent to Alice, likewise, which `$j$' value is sent to Bob is not known to Alice. The task of Bob is to identify the value of $X$ which is sent to Alice. Here non-zero error probability is allowed. In this regard, Alice sends either cbit or a qubit to help Bob. The task can be accomplished if and only if it is always possible to achieve a success probability $p>1/2$ in identifying the value of $X$ no matter which value is sent to Alice or what information is obtained by Bob. It is relatively easier to prove that when Alice is allowed to send only a cbit to Bob, the task cannot be accomplished. However, in the following, we provide stronger observation.

\begin{proposition}
The task can be accomplished if Alice sends a qubit to Bob. Again, it may require high dimensional classical communication for accomplishing the task. Thus, it establishes a large separation between the quantum and the classical communication. 
\end{proposition}

\begin{proof}
For achieving a success probability $p>1/2$, it is necessary to obey the following during encoding. Suppose, a `$j$' value is sent to Bob. Then, corresponding to this `$j$' value we assume that $S_j$ contains $x_i$ and $x_k$. Now, they should be mapped like: $x_i\rightarrow \ket{\phi_i}$ and $x_k\rightarrow \ket{\phi_k}$, where $\ket{\phi_i}$ and $\ket{\phi_k}$ are not same states, i.e., they are linearly independent. Ultimately, Alice should apply an encoding $x_i\rightarrow\ket{\phi_i}$ $\forall i=1,2,\dots d$, such that if we pick an arbitrary pair of states $\ket{\phi_i}$ and $\ket{\phi_k}$ then, they are linearly independent. This is possible since $\ket{0}$ and $\ket{1}$ can be superposed in infinitely many ways. So, if $\ket{\phi_i}$ takes the form $a\ket{0}+b\ket{1}$, then, one can take $\ket{\phi_k}$ as $c\ket{0}+d\ket{1}$, $a,b,c,d$ are positive numbers such that $a^2+b^2=1=c^2+d^2$. We can choose the numbers in a way that $a/b\neq c/d$, which ensures linear independence. In this way, the task can be accomplished if Alice is allowed to send a qubit to Bob, does not matter what is the dimension of $X$. These complete the proof of the first part of the proposition.

For the second part, we first recall that the information which is sent to Bob is not known to Alice and therefore, it is not useful for encoding. We then assume that Alice is allowed to send $(d-1)$-level classical information when the dimension of $X$ is $d$. Clearly, Alice must encode some property associated with the values of $X$ into a function which has no more than $(d-1)$ outputs. If this is the case then there exist at least two values of $X$, corresponding to which the computed function outputs same value. This implies that there exists a `$j$' value, corresponding to which the value of $X$ cannot be identified with a probability better than $1/2$. Thus, for a large value of $d$, it is required to send classical information of high dimension. In fact, it is necessary to communicate $d$-level classical information to achieve the goal of the present task. This directly follows from the fact that communication of $(d-1)$-level classical information from Alice to Bob is not sufficient to achieve the goal of the present task. Now, it is quite straightforward to see that communication of $d$-level classical information is also sufficient to achieve the goal of the present task.

In this way we establish a large gap between the quantum and the classical communication, i.e., quantum communication provides a very large advantage over its classical counterpart when the present task is concerned. Furthermore, theoretically this advantage can be stretched to an arbitrary height.
\end{proof}

From the above proposition, it is clear that quantum communication can provide a very large advantage over its classical counterpart when the present task is concerned. In this context, it is important to mention about Ref.~\cite{Renner23} where classical cost of qubit transmission is discussed in terms of classical communication and shared randomness. So, with respect to such a result, apparently, it is not possible to have large separation between the quantum and the classical communication. But in our case when we consider high dimensional random variable, we do not consider shared randomness along with classical communication like we have done during $d=3$ case. This is the reason why we get a large separation here.

\section{Conclusion}
In the era of developing quantum technologies, advantage of quantum resources and certification of security are two of the most significant aspects of studying quantum information and computation. In this work we have studied both aspects by first setting a unified scenario of a tripartite communication task and then, connecting it with semi-device-independence. 

In particular, we have introduced the following communication task. There are three parties - the Manager, Alice, and Bob. The Manager sends a value of a random variable to Alice and at the same time the Manager also sends some partial information regarding that value to Bob. Bob does not know which value is sent to Alice, likewise, Alice also does not know what information is sent to Bob. The goal of the task can be achieved if and only if the value of the random variable, which is sent to Alice by the Manager, can be identified by Bob with success probability greater than half. Actually, there will be a non-zero error probability. However, to help Bob, Alice sends a limited amount of classical or quantum information to him.

For this task, we have shown the following. When Alice sends a qubit the goal of the task can be achieved. On the other hand, a cbit communication is not sufficient for achieving the goal of the task. Particularly, cbit communication enables Bob to guess certain information sent from the Manager to Alice with probability 1, while guessing other information randomly. In contrast, quantum coherence enables Bob to guess the value of the variable in Alice's hands with success probability always greater than $1/2$. This manifests one of the main differences between the cbit communication and the qubit communication. Moreover, the optimal success probability in the overall process for a qubit communication is higher than the optimal success probability for a cbit communication. In fact, this success probability is higher for qubit communication even if along with cbit communication Alice and Bob share randomness which allows mixing of classical strategies. Thus, it not only establishes quantum advantage, it also exhibit a more prominent non-classical feature, present in this task. Actually, we have obtained a higher success probability compared to other tasks and therefore, this task might be experiment friendly. For applications, we have shown that using present theory it is possible to detect the quantumness of communication in a semi-device-independent manner. We further have provided sufficient condition to detect universal coherence present in an ensemble within the same set up. Finally, we have shown that for a high dimensional random variable, to achieve the goal of the present task, it may require a very high dimensional classical communication. On the other hand, the goal can be achieved via only a qubit communication. This establishes a very large separation between the quantum and the classical communication. 

For further research, it will be interesting if it is possible to analyze our findings in light of the ``information causality principle'' \cite{Pawowski09}. Furthermore, it will also be interesting if it is possible to demonstrate our task experimentally. In this regard, we mention that we have used Helstrom bound in many places to obtain optimal figures of merit and this bound is known to be achievable in a practical scenario.  

\section*{Acknowledgments}
This work was supported by the ``Quantum Optical Technologies'' project, carried out within the International Research Agendas programme of the Foundation for Polish Science co-financed by the European Union under the European Regional Development Fund, and the National Science Centre Poland within the QuantERA II Programme (Grant No.~2021/03/Y/ST2/00178, acronym ExTRaQT) that has received funding from the European Union's Horizon 2020 research and innovation programme under Grant Agreement No.~101017733.

\bibliographystyle{apsrev4-2}
\bibliography{ref}
\end{document}